\documentclass[english]{IEEEtran}
\usepackage[T1]{fontenc}
\usepackage[latin9]{inputenc}
\usepackage{verbatim}
\usepackage{amsthm}
\usepackage{amsmath}
\usepackage{graphicx}
\usepackage{amssymb}
\usepackage{esint}

\makeatletter

\newcommand{\lyxline}[1][1pt]{%
  \par\noindent%
  \rule[.5ex]{\linewidth}{#1}\par}

  \theoremstyle{plain}
  \newtheorem{lem}{Lemma}
\theoremstyle{plain}
\newtheorem{thm}{Theorem}
  \theoremstyle{plain}
  \newtheorem{cor}{Corollary}

\usepackage{cite}
\usepackage{wrapfig}
\IEEEoverridecommandlockouts

\makeatother

\usepackage{babel}

\begin{document}

\title{\vspace{-0.1in}
Downlink Coverage Analysis in a Heterogeneous Cellular Network}

\author{Prasanna Madhusudhanan, Juan G. Restrepo, Youjian (Eugene) Liu, Timothy X Brown \thanks{P. Madhusudhanan, Y. Liu, and T. X. Brown are with the Department of Electrical, Computer and Energy Engineering; J. G. Restrepo is with the Department of Applied Mathematics; T. X. Brown is also with the Interdisciplinary Telecommunications Program, at the University of Colorado, Boulder, CO 80309-0425 USA. Email: \{mprasanna, juanga, eugeneliu, timxb\}@colorado.edu}\vspace{-0.4in}
}
\maketitle
\begin{abstract}
In this paper, we consider the downlink signal-to-interference-plus-noise
ratio (SINR) analysis in a heterogeneous cellular network with \emph{K}
tiers. Each tier is characterized by a base-station (BS) arrangement
according to a homogeneous Poisson point process with certain BS density,
transmission power, random shadow fading factors with arbitrary distribution,
arbitrary path-loss exponent and a certain bias towards admitting
the mobile-station (MS). The MS associates with the BS that has the
maximum SINR under the open access cell association scheme. For such
a general setting, we provide an analytical characterization of the
coverage probability at the MS. \end{abstract}

\section{Introduction\label{sec:Introduction}}

\IEEEPARstart{T}{he} heterogeneous cellular network is a complex
overlay of multiple cellular communication networks such as the macrocells,
microcells, picocells, and femtocells. Research has shown heterogeneous
networks support greater end-user data-rate and throughput as well
as better indoor and cell-edge coverage. This has further led to its
inclusion as an important feature for the 4G cellular networks \cite{Report2011,Qualcomm2010,Chandrasekhar2008,Lagrange1997}.

In the conception of the heterogeneous cellular network, one is looking
at an overlay of several dense, irregularly and often completely randomly
deployed networks (namely microcells, picocells and femtocells) with
a limited coverage area, all deployed on top of the conventional macrocell
network. These networks consist of base-stations (BSs) with different
transmission powers, different traffic-load carrying capability, and
different radio environment which is based on the locations in which
they are deployed. All these sum up to an extremely complicated network.
As a result, the analysis of such a network through system simulations
(which is largely the approach taken for studying the conventional
macrocell network) is hampered by the curse of dimensionality due
to the many parameters involved in designing and modeling each of
the representative networks that make up the heterogeneous network.
For this reason, we seek to develop an analytical model that captures
all the design scenarios of interest. 

The analysis in this paper applies to the downlink performance in
terms of the signal-to-interference-plus-noise ratio (SINR) at the
mobile-station (MS), where the MS associates itself with the BS that
has the maximum SINR at the MS. The SINR at the MS is an important
metric that determines the outage probability (or coverage probability),
capacity and throughput of a cellular network in the downlink, and
the characterization of the distribution of SINR aids in the complete
understanding of the SINR metric.

The recent focus on stochastic geometry as a means to study and analyze
large systems of essentially randomly deployed nodes, combined with
the fact that certain representative networks (femtocell networks),
that are a part of the heterogeneous network, are formed due to the
end-user deployments, thus falling under the random deployment category,
has naturally led to the application of stochastic geometry in modeling
the heterogeneous network. While there is a vast literature corresponding
to stochastic geometric modeling and analysis of systems with randomly
deployed nodes, there have been three independent efforts to applying
stochastic geometry to study heterogeneous networks \cite{Dhillon2011,Mukherjee2012,Madhusudhanan2011}.
All the 3 papers modeled the heterogeneous network as being composed
of multiple tiers where the BSs in each tier are deployed according
to a homogeneous Poisson point process, independent of the other tiers.
While the first two references focused on the case where the fading
coefficients were modeled as independent and identically distributed
(i.i.d.) random variables with exponential distribution and obtained
closed-form expressions for the distribution of SINR at the MS in
a heterogeneous network for positive values of SINR (in dB), our work
has provided semi-analytical expressions for the same quantities and
holds for arbitrary fading distributions and for all values of SINR.

Jo et. al. \cite{Jo2011} have extended their results to obtain a
complete characterization of the SINR in the heterogeneous network
where the exponents for the power-law path-loss model were different
for different tiers of the heterogeneous network, but for the case
when the MS associates itself with the nearest BS rather than the
BS that maximizes the SINR, and when an exponential distribution was
assumed for fading. With this they could study the effect of varying
cell-association biasing.

In this paper, we incorporate all the modeling details available in
the literature for studying the heterogeneous cellular network, introduce
a few additional important features, and obtain the complete characterization
of the downlink SINR, and the coverage probability (i.e. 1 - outage
probability) in a heterogeneous network, for the original case where
the MS associates with the BS with the maximum SINR at the MS. In
particular, the following paragraph lists our contributions.

From the modeling stand-point, the transmission and channel characteristics
of a BS corresponding to a given tier include the transmission power,
bias factor, path-loss exponent, distribution of fading coefficients,
and the SINR threshold to be satisfied by the MS in order to communicate
with the BS. These may be different for different tiers. This is more
general than the setting in the previous heterogeneous network analysis
literature. For this situation, we have obtained accurate characterization
for the SINR distribution as well as the coverage probability for
a given MS. Since we are able to handle arbitrary fading distributions,
that are further different for different tiers, we are able to consider
practical fading models and obtain the coverage probabilities in those
cases. Such a strong result is likely to be useful for studying and
analyzing realistic scenarios. The following section introduces the
system model.

\section{System Model\label{sec:modelreview}}

This section describes the various elements used to model the wireless
network, namely, the BS layout, the radio environment, and the performance
metrics of interest.

\subsubsection{BS Layout\label{sub:BS-Layout}}

The BS layout for the $k^{\mathrm{th}}$ tier, where $k\in\left\{ 1,\dots,\ K\right\} ,$
is according to independent homogeneous Poisson point process with
density, $\lambda_{k}.$

\subsubsection{Cell-Association policy\label{sub:Cell-Association-policy}}

The MS associates itself to the BS corresponding to the the strongest
instantaneous received power (the BS from which the MS has the maximum
$SINR$). We focus on the open access scheme in this paper where the
MS can freely communicate with any of the $K$ tiers.

\subsubsection{Radio Environment\label{sub:Radio-Environment}}

The received power at the MS from the $j^{th}$ BS belonging to the
$k^{\mathrm{th}}$ tier at a distance $D_{kj}\ (>0)$ from the MS
is given by $P=P_{k}\Psi_{kj}D_{kj}^{-\varepsilon_{k}}B_{k},$ where
$\left(P_{k},\ \Psi_{kj},\ B_{k},\ \varepsilon_{k}\right)$ corresponds
to the constant transmission power, random channel gain coefficient,
constant bias coefficient and the constant path-loss exponent ($>2$)
of the $k^{\mathrm{th}}$ tier, respectively. Further, $\Psi_{kj}$
can assume any arbitrary distribution as long as $\mathbb{E}\left[\Psi_{kj}^{2\left/\varepsilon_{k}\right.}\right]<\infty.$
It is independent and identically distributed (i.i.d.) across all
the BSs of the $k^{\mathrm{th}}$ tier, and independent of the other
tiers and the underlying random process governing the BS arrangement.
For this reason, we drop the subscript $j$ and for compactness denote
the expectation by $\mathbb{E}\Psi_{k}^{\frac{2}{\varepsilon_{k}}}.$
In essence, the different tiers have their own arbitrary fading distributions.
All BSs of the $k^{\mathrm{th}}$ tier adopt an identical bias factor,
$B_{k},$ which is some positive value. By manipulating the $B_{k}$'s
corresponding to the different tiers, we may regulate the traffic
from one tier to the other.

\subsubsection{Performance Metric\label{sub:Performance-Metric}}

In this paper, we are concerned with the SINR at a given MS. Without
loss of generality, the MS is assumed to be located at the origin
of the plane. SINR is defined as the ratio of the received signal
from the desired BS to the sum of the interferences from all the BSs
belonging to all the $K$ tiers and the background noise. As a result
\begin{eqnarray}
SINR & = & \frac{P_{k}\Psi_{kj}R_{kj}^{-\varepsilon_{k}}B_{k}}{\underset{\left(m,l\right)\ne\left(k,\ j\right)}{\sum_{m=1}^{K}\sum_{l=1}^{\infty}}P_{m}\Psi_{ml}R_{ml}^{-\varepsilon_{m}}B_{m}+\eta},\label{eq:SINRDefinition}\end{eqnarray}
where $\left(k,\ j\right)$ corresponds to the indices of the tier
and the corresponding BS, which has the maximum received power at
the MS, and $\eta$ is the power corresponding to the background noise.
Further, each user can successfully communicate with its desired BS
provided the $SINR$ is above the minimum threshold value, $\beta_{k},$
that is a characteristic of the tier. As a result, the coverage probability
is defined as the probability that the MS is able to communicate with
a BS using a specific BS association protocol.

\section{Useful Lemmas\label{sec:SINRcharacteristics}}

We will first present some lemmas which will be useful in deriving
the coverage probability.
\begin{lem}
\textup{\label{lem:StochasticEqLemma}The SINR at the MS has the same
distribution as that of a single-tier network where all BSs in the
network have unity transmission power, channel gain and path-loss
exponent, and are arranged according to a non-homogeneous 1-D Poisson
point process with BS density function $\lambda\left(r\right)=\sum_{l=1}^{K}\lambda_{l}\left(P_{l}B_{l}\right)^{\frac{2}{\varepsilon_{l}}}\mathbb{E}\Psi_{l}^{\frac{2}{\varepsilon_{l}}}r^{\frac{2}{\varepsilon_{l}}-1},$
$r\ge0,$ as long as $\mathbb{E}\Psi_{l}^{\frac{2}{\varepsilon_{l}}}<\infty,\ \forall\ l=1,\ 2,\cdots,\ K.$
That is} \begin{eqnarray}
\, & \, & \,\nonumber \\
 & SINR & =_{\mathrm{st}}\left.\frac{\tilde{R}_{1}^{-1}}{\sum_{i=2}^{\infty}\tilde{R}_{i}^{-1}+\eta}\right|_{\lambda\left(r\right)},\label{eq:SINRDistributionEq}\end{eqnarray}
\textup{ where $=_{\mathrm{st}}$ indicates the equivalence in distribution;
and $\left\{ \tilde{R}_{i}\right\} _{i=1}^{\infty}$ is the ascendingly
ordered distances of the BSs from the origin, obtained from a non-homogeneous
1-D Poisson point process with BS density function $\lambda\left(r\right)$
defined above.}\end{lem}
\begin{proof}
See Appendix \ref{sub:proofStEquivalenceLemma}.
\end{proof}
Now, for the equivalent single-tier network, we characterize the distance
of the nearest $k^{\mathrm{th}}$ tier BS from the MS, $\tilde{R}_{1}^{\left(k\right)}.$
\begin{lem}
\emph{\label{thm:nearestKthTierBSDistance} The tail probability of
$\tilde{R}_{1}^{\left(k\right)}$ is }\begin{eqnarray}
 &  & \mathbb{P}\left(\left\{ \tilde{R}_{1}^{\left(k\right)}>r\right\} \right)\nonumber \\
 &  & =\exp\left(-\lambda_{k}\pi\left(P_{k}B_{k}\right)^{\frac{2}{\varepsilon_{k}}}\mathbb{E}\Psi_{k}^{\frac{2}{\varepsilon_{k}}}r^{\frac{2}{\varepsilon_{k}}}\right),\ \forall\ r\ge0.\label{eq:nearestKthTierBSDistance}\end{eqnarray}
\end{lem}
\begin{proof}
See Appendix \ref{sub:proofRkTailProbTheorem}.
\end{proof}
Using the above result, we will now characterize the random variable
$I$, which represents the tier to which the desired BS (i.e., the
BS nearest to the MS in the equivalent single-tier network mentioned
in Lemma \ref{lem:StochasticEqLemma}) belongs.
\begin{lem}
\emph{\label{lem:TierProbabilityLemma}The desired BS belongs to the
$k^{\mathrm{th}}$ tier $\left(k=1,\ 2,\cdots,\ K\right)$ with the
probability \begin{eqnarray}
 &  & \mathbb{P}\left(\left\{ I=k\right\} \right)=\int_{t=0}^{\infty}\lambda_{k}\mathbb{E}\Psi_{k}^{\frac{2}{\varepsilon_{k}}}2\pi t\times\nonumber \\
 &  & \exp\left(-\sum_{l=1}^{K}\lambda_{l}\mathbb{E}\Psi_{l}^{\frac{2}{\varepsilon_{l}}}\pi\left(\frac{P_{l}B_{l}}{P_{k}B_{k}}\right)^{\frac{2}{\varepsilon_{l}}}t^{\frac{2\varepsilon_{k}}{\varepsilon_{l}}}\right)dt.\label{eq:TierProbabilityGeneral}\end{eqnarray}
Further, in the special case $\left\{ \varepsilon_{k}\right\} _{k=1}^{K}=\varepsilon,$
\begin{eqnarray}
\mathbb{P}\left(\left\{ I=k\right\} \right) & = & \frac{\lambda_{k}\mathbb{E}\Psi_{k}^{\frac{2}{\varepsilon}}\left(P_{k}B_{k}\right)^{\frac{2}{\varepsilon}}}{\sum_{m=1}^{K}\lambda_{m}\mathbb{E}\Psi_{m}^{\frac{2}{\varepsilon}}\left(P_{m}B_{m}\right)^{\frac{2}{\varepsilon}}}.\label{eq:TierProbabilitySpecialCase}\end{eqnarray}
}%
\begin{comment}
\emph{It is a function of \textbackslash{}lambda\_m. What do you mean
it does not depend? Sir, you are right! I didn't mean that. I have
made the correction.}
\end{comment}
{}\end{lem}
\begin{proof}
See Appendix \ref{sub:proofDesiredBSTierProbLemma}.
\end{proof}
Although we do not have a closed-form expression, $\left(\ref{eq:TierProbabilityGeneral}\right)$
can be computed easily to any desired accuracy by numerical integration.
Next, we obtain the p.d.f. of the distance of the serving BS from
MS, given that it belongs to the $k^{\mathrm{th}}$ tier, denoted
by $f_{\left.\tilde{R}_{1}\right|I}\left(\left.r\right|k\right).$
\begin{lem}
\emph{\label{lem:pdfServingBSGivenTier}The p.d.f. of the distance
of the serving BS from the MS, given it belongs to the $k^{\mathrm{th}}$
tier }$\left(k=1,\ 2,\cdots,\ K\right),$\emph{ is} \begin{eqnarray}
 &  & f_{\left.\tilde{R}_{1}\right|I}\left(\left.r\right|k\right)=\frac{\lambda_{k}\frac{2\pi}{\varepsilon_{k}}\left(P_{k}B_{k}\right)^{\frac{2}{\varepsilon_{k}}}\mathbb{E}\Psi_{k}^{\frac{2}{\varepsilon_{k}}}r^{\frac{2}{\varepsilon_{k}}-1}}{\mathbb{P}\left(\left\{ I=k\right\} \right)}\nonumber \\
 &  & \times\exp\left(-\sum_{l=1}^{K}\lambda_{l}\pi\left(P_{l}B_{l}\right)^{\frac{2}{\varepsilon_{l}}}\mathbb{E}\Psi_{l}^{\frac{2}{\varepsilon_{l}}}r^{\frac{2}{\varepsilon_{l}}}\right),\ r\ge0,\label{eq:pdfServingBSgivenTier}\end{eqnarray}
\emph{where the denominator is obtained using Lemma \ref{lem:TierProbabilityLemma}.}\end{lem}
\begin{proof}
See Appendix \ref{sub:proofPDFRkGivenTier}.
\end{proof}
Now, we are ready to obtain the expression for the coverage probability
for a typical MS in this heterogeneous network.

\section{Coverage Probability}

Recall that the MS is covered only if the SINR exceeds a certain threshold,
$\beta_{k},$ where $k$ is the tier to which the desired BS belongs.
Using the results from the previous section, we now present the coverage
probability for an MS in a heterogeneous cellular network.
\begin{thm}
\label{thm:coverageProbabilityGeneral}\emph{The coverage probability
of the MS is \begin{eqnarray}
 &  & \mathbb{P}_{coverage}^{open-access}=\nonumber \\
 &  & \sum_{k=1}^{K}\lambda_{k}\mathbb{E}\Psi_{k}^{\frac{2}{\varepsilon_{k}}}\cdot\int_{\omega=-\infty}^{\infty}\int_{t=0}^{\infty}\mathrm{e}^{\frac{i\omega\eta\left(t\left/\pi\right.\right){}^{\frac{\varepsilon_{k}}{2}}}{P_{k}B_{k}}}\nonumber \\
 &  & \mathrm{e}^{-\sum_{l=1}^{K}\lambda_{l}\alpha_{l}\mathbb{E}\Psi_{l}^{\frac{2}{\varepsilon_{l}}}\left(\frac{P_{l}B_{l}}{P_{k}B_{k}}\right)^{\frac{2}{\varepsilon_{l}}}t^{\frac{\varepsilon_{k}}{\varepsilon_{l}}}}dt\cdot\left(\frac{1-\mathrm{e}^{\frac{-i\omega}{\beta_{k}}}}{2\pi i\omega}\right)d\omega,\label{eq:coverageProbabilityGeneral}\end{eqnarray}
where $\alpha_{l}=\pi^{1-\frac{\varepsilon_{k}}{\varepsilon_{l}}}\cdot_{1}F_{1}\left(\frac{-2}{\varepsilon_{l}};1-\frac{2}{\varepsilon_{l}};i\omega\right).$}\end{thm}
\begin{proof}
See Appendix \ref{sub:ProofCovProbTheorem}.
\end{proof}
Though the above expression is not in closed form, this can also be
computed to any desired accuracy using numerical integration methods
to compute the double integral.%
\begin{comment}
how? use a few words to say.
\end{comment}
{} Moreover, several insightful results arise for certain special cases
of the above result. For example, in the interference-limited scenario,
where the effect of the background noise may be ignored in the presence
of the strong interferences from the BSs, the above result reduces
to a simple form as shown below. 
\begin{cor}
\emph{\label{cor:samePLexponentCase}In the interference-limited case,
with $\left\{ \varepsilon_{k}\right\} _{k=1}^{K}=\varepsilon,$ the
coverage probability is} \begin{eqnarray*}
\mathbb{P}_{coverage}^{open-access} & = & \sum_{k=1}^{K}\frac{\lambda_{k}\mathbb{E}\Psi_{k}^{\frac{2}{\varepsilon}}\left(P_{k}B_{k}\right)^{\frac{2}{\varepsilon}}\gamma_{k}}{\sum_{l=1}^{K}\lambda_{l}\mathbb{E}\Psi_{l}^{\frac{2}{\varepsilon}}\left(P_{l}B_{l}\right)^{\frac{2}{\varepsilon}}},\end{eqnarray*}
\emph{where $\gamma_{k}=\int_{\omega=-\infty}^{\infty}\frac{\left(1-\mathrm{e}^{-\frac{i\omega}{\beta_{k}}}\right)}{2\pi i\omega\cdot_{1}F_{1}\left(-\frac{2}{\varepsilon};1-\frac{2}{\varepsilon};i\omega\right)}d\omega,$
and this integration reduces to $\gamma_{k}=\frac{\mathrm{sin}\left(\left.2\pi\right/\varepsilon\right)}{\left.2\pi\right/\varepsilon}\cdot\beta_{k}^{-\frac{2}{\varepsilon}},$
if }\textbf{\emph{$\beta_{k}\ge1.$}}
\end{cor}
The detailed proof will not be given. However, the corollary can be
easily proved by noting that the integration w.r.t. $t$ in $\left(\ref{eq:coverageProbabilityGeneral}\right)$
evaluates to $\frac{\left(\sum_{l=1}^{K}\lambda_{l}\mathbb{E}\Psi_{l}^{\frac{2}{\varepsilon}}\left(\frac{P_{l}B_{l}}{P_{k}B_{k}}\right)^{\frac{2}{\varepsilon}}\right)^{-1}}{_{1}F_{1}\left(-\frac{2}{\varepsilon};1-\frac{2}{\varepsilon};i\omega\right)}.$
Notice that $\gamma_{k}$ does not depend on any of the parameters
that define the characteristics of the BSs of the various tiers, and
only depends on the path-loss exponent, and the SINR thresholds of
the various tiers. Further, $\gamma_{k}$ is the coverage probability
of a MS in a single-tier network. For the $\eta=0$ case, \cite[Eq. 3]{Dhillon2011}
gives the coverage probability of a single-tier network with SINR
threshold, $\gamma\ \left(\ge1\right),$ and channel gains that are
i.i.d. exponentially distributed to be $\frac{\sin\left(2\pi/\varepsilon\right)}{\left(2\pi/\varepsilon\right)}\cdot\gamma^{-\frac{2}{\varepsilon}}.$
For the same case, in \cite[Remark 4]{Madhusudhanan2010a}, we have
shown that the single-tier network $SINR$ distribution is the same
irrespective of the transmission power of BSs and the distribution
of the channel gains. As a result, the above expression that holds
for exponential fading distribution also holds for any other fading
distribution (even no fading). 

The contributions of the BS density, transmission power, fading distribution,
and the bias factor of the $k^{\mathrm{th}}$ tier BS, are all captured
as the multiplicative factors of the $k^{\mathrm{th}}$ tier SINR
threshold, $\gamma_{k}.$ Further, the contribution of the fading
coefficient is completely captured in terms of the $\frac{2}{\varepsilon_{k}}^{\mathrm{th}}$
moment of the random variable $\Psi_{k},$ used to represent the shadow
fading factor of the $k^{\mathrm{th}}$ tier. Next, having studied
the coverage probability of the MS in the heterogeneous cellular network,
we present some numerical examples to illustrate them.

\section{\label{sec:NumericalExamplesAndDiscussion}Numerical Examples and
Discussion}

In this section, we study various scenarios in order to clearly illustrate
the results we have obtained. We restrict ourselves to a two-tier
network consisting of a macrocell and picocell network for simplicity,
and assume that the background noise is zero. We note that these studies
can be extended to arbitrary number of tiers. Further, please refer
Appendix \ref{sub:Simulation-Method} for the algorithm used to perform
the Monte-Carlo simulations. For all the cases that will be considered
next, we assume $\lambda_{1}=0.001,$ $\lambda_{2}=.002,$ $P_{1}=53\ \mathrm{dBm},$
$P_{2}=33\ \mathrm{dBm,}$ $B_{1}=1,$ $B_{2}=1,$ $\Psi_{1}$ and
$\Psi_{2}$ are both exponential random variables with mean 1, $\varepsilon_{1}=3.8,$
$\varepsilon_{2}=3.5,$ \textbf{$\beta_{1}=\beta_{2}=0\ \mathrm{dB},$}
unless specified other wise. %
\begin{figure}
\begin{centering}
\vspace{-0.1in}
\includegraphics[clip,scale=0.6]{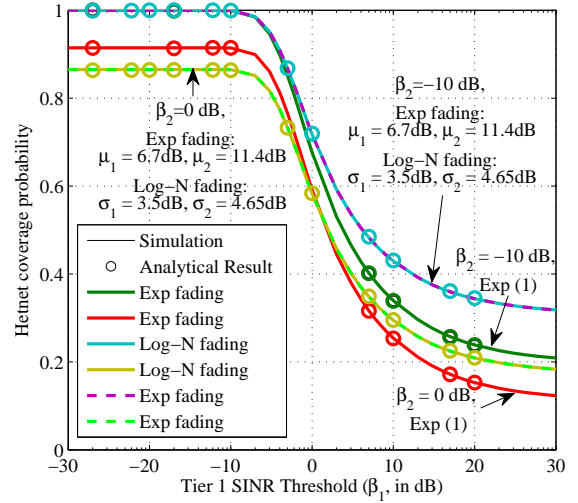}
\par\end{centering}

\caption{\label{fig:covProbVSTier1SINRTh}2-Tier network, Coverage Probability
vs Tier 1 SINR threshold}
\vspace{-0.1in}

\end{figure}
\begin{figure}
\begin{centering}
\vspace{-0.1in}
\includegraphics[clip,scale=0.6]{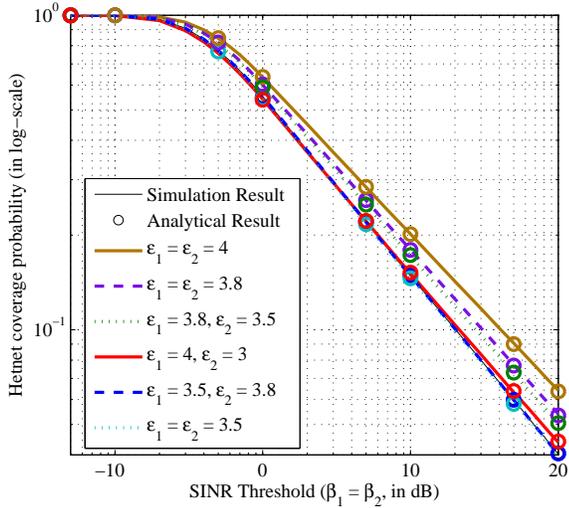}
\par\end{centering}

\caption{\label{fig:2TierNwDiffPLConfigs}2-Tier network, Coverage Probability
vs SINR threshold for various combinations of path-loss exponents}
\vspace{-0.1in}

\end{figure}
\begin{figure}
\begin{centering}
\vspace{-0.1in}
\includegraphics[clip,scale=0.6]{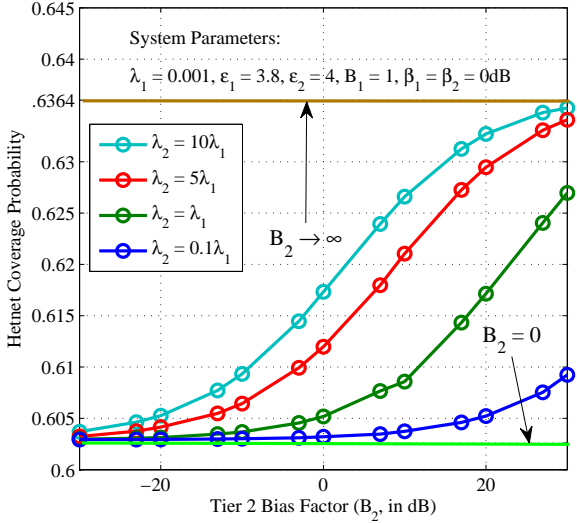}
\par\end{centering}

\caption{\label{fig:2TierNWCovProbVSTier2Bias}2-Tier network, Coverage probability
vs Tier 2 Bias Factor}
\vspace{-0.1in}

\end{figure}

In our first example, we dwell in detail into our characterization
of the SINR distribution as well as the coverage probability for arbitrary
fading distribution that are different for different tiers. Figure
\ref{fig:covProbVSTier1SINRTh} shows the plot of coverage probability
versus the SINR threshold for various choices of fading distributions
for each of the tiers. The first two curves (in the legend of Figure
\ref{fig:covProbVSTier1SINRTh}) show the two-tier network coverage
probability when the channel gains at both the tiers are exponential
random variables with mean 1. In the next two curves in Figure \ref{fig:covProbVSTier1SINRTh},
we depict two scenarios where the channel gains are log-normal random
variables with zero mean and standard deviations $\sigma_{1}=3.5\ \mathrm{dB}$
for the first tier and $\sigma_{2}=4.65\ \mathrm{dB}$ for the second
tier. The characterization of the coverage probability if $\beta_{1}\ne\beta_{2}$
and if \textbf{$\beta_{1},\ \beta_{2}\ngeq0\ \mathrm{dB}$ }were not
known until now, for any chosen distribution of the channel gains,
not even the exponential distribution.\textbf{ }The results in this
paper, in particular, Theorem \ref{thm:coverageProbabilityGeneral}
gives the coverage probability for all values of $\beta_{1}$ and
$\beta_{2}.$ The last two curves in Figure \ref{fig:covProbVSTier1SINRTh}
consider the channel gain to have exponential distributions with means
$\mu_{1}=46.5\ \mathrm{dB}$ for the first tier and $\mu_{2}=50.3\ \mathrm{dB}$
for the second tier. Notice that their coverage probability curves
match exactly with those for the log-normal distributions. This is
because, as illustrated in Theorem \ref{thm:coverageProbabilityGeneral},
the coverage probability only depends on the $\frac{2}{\varepsilon_{k}}^{\mathrm{th}}$
moment of the random fading factor of the $k^{\mathrm{th}}$ tier,
i.e. $\mathbb{E}\Psi_{k}^{\frac{2}{\varepsilon_{k}}}.$ If $\Psi_{k}$
is a log-normal random variable with standard deviation $\sigma_{k},$
then $\mathbb{E}\Psi_{k}^{\frac{2}{\varepsilon}}=\exp\left(\frac{2\sigma_{k}^{2}}{\varepsilon_{k}^{2}}\right),$
and if $\Psi_{k}$ is an exponential random variable with mean $\mu_{k},$
the $\mathbb{E}\Psi_{k}^{\frac{2}{\varepsilon_{k}}}=\mu_{k}^{\frac{2}{\varepsilon_{k}}}\Gamma\left(1+\frac{2}{\varepsilon_{k}}\right).$
In this example, we have chosen $\sigma_{1},\ \sigma_{2},\ \mu_{1},$
and $\mu_{2}$ in such a way that $\mathbb{E}\Psi_{k}^{\frac{2}{\varepsilon_{k}}}$'s
are the same for $k=1,\ 2.$

In the next example, we show that we are now able to study the heterogeneous
network for different path-loss exponents at different tiers. Notice
that there is an improvement when the path-loss exponents are large.
This is because the signal power decays faster with the distance,
thereby causing lesser intercell interference. Notice from Corollary
\ref{cor:samePLexponentCase} that, when the path-loss exponents and
the SINR thresholds are identical across the tiers, then the coverage
probability is the same as that of a single-tier network with the
same path-loss exponent and the SINR threshold, and further, coverage
probability varies log-linearly with the SINR threshold (> 0 dB).
In Figure \ref{fig:2TierNwDiffPLConfigs}, notice that even when the
path-loss exponents are not identical across the tiers, the coverage
probability still has a log-linear behavior for SINR thresholds greater
than 0 dB.

In the last example, we study the effect of varying the BS density
$\left(\lambda_{2}\right)$ and the bias factor $\left(B_{2}\right)$
of the second tier on the coverage probability, where the second tier
has a greater path-loss exponent, to mimick a typical indoor environment
situation. While Corollary \ref{cor:samePLexponentCase} shows that
the coverage probability does not depend on the BS density of the
tiers as well as the bias factors of the tiers when all the tiers
have the same path-loss exponent and the SINR thresholds, when the
path-loss exponents are different, Figure \ref{fig:2TierNWCovProbVSTier2Bias}
shows that the coverage probability actually improves as we increase
the bias factor of the second tier, and further it increases at a
faster rate as the density of the second tier is increased. At the
limits of the tier two bias factor, i.e. \textbf{$B_{2}\rightarrow0,$
}and $B_{2}\rightarrow\infty,$ the two-tier network essentially collapses
to a single-tier network consisting of only the macrocell network,
and the femtocell network, respectively. As mentioned previously,
the single tier network is invariant to changes in the BS density
and the bias factor, and as a result the curves are straight lines
at these limits. 

The mathematical tools developed in this paper to study the heterogeneous
network coverage probability are sufficient to characterize the average
ergodic rate achieved at the MS, throughput and the per-tier traffic
load, which are other important metrics of interest for understanding
the heterogeneous network, and these will be considered in detail
elsewhere.

\section{Conclusions}

In this paper, we study a heterogeneous cellular network consisting
of $K$ tiers, where each tier has its own BS density, BS transmission
power and bias factor, path-loss exponent and channel gain with an
arbitrary distribution, that is different for different tiers. For
such a general model for the heterogeneous network, we develop mathematical
tools based on stochastic geometry to characterize the distribution
of the downlink SINR and the coverage probability at any given MS,
where the MS associates itself with the BS that has the SINR at the
MS. Moreover, we have achieved a complete characterization of the
SINR distribution and the coverage probability for all values of SINR
thresholds, which has not been done before.

\bibliographystyle{IEEEtran}
\bibliography{prasanna}

\begin{thebibliography}{10}
\providecommand{\url}[1]{#1}
\csname url@samestyle\endcsname
\providecommand{\newblock}{\relax}
\providecommand{\bibinfo}[2]{#2}
\providecommand{\BIBentrySTDinterwordspacing}{\spaceskip=0pt\relax}
\providecommand{\BIBentryALTinterwordstretchfactor}{4}
\providecommand{\BIBentryALTinterwordspacing}{\spaceskip=\fontdimen2\font plus
\BIBentryALTinterwordstretchfactor\fontdimen3\font minus
  \fontdimen4\font\relax}
\providecommand{\BIBforeignlanguage}[2]{{%
\expandafter\ifx\csname l@#1\endcsname\relax
\typeout{** WARNING: IEEEtran.bst: No hyphenation pattern has been}%
\typeout{** loaded for the language `#1'. Using the pattern for}%
\typeout{** the default language instead.}%
\else
\language=\csname l@#1\endcsname
\fi
#2}}
\providecommand{\BIBdecl}{\relax}
\BIBdecl

\bibitem{Report2011}
\BIBentryALTinterwordspacing
{4G Americas Report}. (2011, February) 4{G} {M}obile {B}roadband {E}volution:
  3{GPP} {R}elease 10 and {B}eyond. [Online]. Available:
  \url{http://www.4gamericas.org/}
\BIBentrySTDinterwordspacing

\bibitem{Qualcomm2010}
\BIBentryALTinterwordspacing
Qualcomm. (2010, February) Lte advanced: Heterogeneous network. [Online].
  Available:
  \url{http://www.qualcomm.com/documents/files/lte-advanced-heterogeneous-networks.pdf}
\BIBentrySTDinterwordspacing

\bibitem{Chandrasekhar2008}
V.~Chandrasekhar, J.~Andrews, and A.~Gatherer, ``Femtocell {N}etworks: A
  {S}urvey,'' \emph{Communications Magazine, IEEE}, vol.~46, no.~9, pp. 59--67,
  September 2008.

\bibitem{Lagrange1997}
X.~Lagrange, ``Multitier cell design,'' \emph{Communications Magazine, IEEE},
  vol.~35, no.~8, pp. 60 --64, aug 1997.

\bibitem{Dhillon2011}
H.~S. Dhillon, R.~K. Ganti, F.~Baccelli, and J.~G. Andrews, ``Modeling and
  analysis of {K}-tier downlink heterogeneous cellular networks,'' \emph{CoRR},
  vol. abs/1103.2177, 2011.

\bibitem{Mukherjee2012}
S.~Mukherjee, ``Distribution of downlink {SINR} in heterogeneous cellular
  networks,'' \emph{Selected Areas in Communications, IEEE Journal on},
  vol.~30, no.~3, pp. 575 --585, april 2012.

\bibitem{Madhusudhanan2011}
P.~Madhusudhanan, J.~G. Restrepo, Y.~Liu, T.~X. Brown, and K.~Baker,
  ``Multi-tier network performance analysis using a shotgun cellular system,''
  in \emph{IEEE Globecom 2011 Wireless Communications Symposium}, Dec. 2011,
  pp. 1 --6.

\bibitem{Jo2011}
\BIBentryALTinterwordspacing
H.-S. Jo, Y.~J. Sang, P.~Xia, and J.~G. Andrews, ``Heterogeneous cellular
  networks with flexible cell association: A comprehensive downlink sinr
  analysis,'' \emph{CoRR}, vol. abs/1107.3602, 2011. [Online]. Available:
  \url{http://arxiv.org/abs/1107.3602}
\BIBentrySTDinterwordspacing

\bibitem{Madhusudhanan2010a}
\BIBentryALTinterwordspacing
P.~Madhusudhanan, J.~G. Restrepo, Y.~Liu, T.~X. Brown, and K.~Baker,
  ``Generalized carrier to interference ratio analysis for the shotgun cellular
  system,'' \emph{CoRR}, 2010. [Online]. Available:
  \url{http://arxiv.org/abs/1002.3943}
\BIBentrySTDinterwordspacing

\bibitem{Kingman1993}
J.~F.~C. Kingman, \emph{Poisson Processes (Oxford Studies in
  Probability)}.\hskip 1em plus 0.5em minus 0.4em\relax Oxford University
  Press, USA, January 1993.

\end{thebibliography}

\appendix

\subsection{\label{sub:proofStEquivalenceLemma}Proof for the Stochastic Equivalence
Lemma}

Given a BS belonging to the $k^{\mathrm{th}}$ tier is at a distance
$R_{k}$ from the origin, then, $\left.\tilde{R}\right|k=\left(P_{k}B_{k}\Psi_{k}\right)^{-1}R_{k}^{\varepsilon_{k}}$
represents the distance of the BS from the origin where the BS arrangement
is according to a non-homogeneous 1-D Poisson point process with BS
density function $\lambda^{\left(k\right)}\left(r\right)=\lambda_{k}\frac{2\pi}{\varepsilon_{k}}\left(P_{k}B_{k}\right)^{\frac{2}{\varepsilon_{k}}}\mathbb{E}\Psi_{k}^{\frac{2}{\varepsilon_{k}}}r^{\frac{2}{\varepsilon_{k}}-1},\ r\ge0,$
as long as $\mathbb{E}\Psi_{k}^{\frac{2}{\varepsilon_{k}}}<\infty,$
for each $k=1,\ 2,\cdots,\ K.$ This is a consequence of the Mapping
theorem \cite[Page 18]{Kingman1993} and the Marking Theorem \cite[Page 55]{Kingman1993}
of the Poisson processes. Further, since the BS arrangements in the
different tiers were originally independent of each other, the set
of all the BSs in the equivalent 1-D non-homogeneous Poisson process
is merely the union of all $\left.\tilde{R}'s\right|k,\ \forall\ k=1,\ 2,\cdots,\ K.$
By the Superposition Theorem \cite[Page 16]{Kingman1993} of Poisson
process, $\tilde{R}$ (notice that it is not conditioned on $k$)
corresponds to the distance from origin of BS arrangement according
to non-homogeneous Poisson point process with density function $\lambda\left(r\right)=\sum_{k=1}^{K}\lambda^{\left(k\right)}\left(r\right),\ r\ge0.$ 

In summary, we have converted the BS arrangement on a 2-D plane of
a $K-$tier network to a BS arrangement of the equivalent single-tier
network along 1-D (positive x-axis), and further the SINR distributions
of both these networks are also equivalent. Further, by our construction,
the BS corresponding to the strongest received power at the MS, in
the $K-$tier network, corresponds to the BS that is nearest to the
origin (MS) in the equivalent single tier network. As a result, $SINR$
may be written in terms of the $\tilde{R}$'s indexed in the ascending
order, and we get $\left(\ref{eq:SINRDistributionEq}\right).$

\subsection{\label{sub:proofRkTailProbTheorem}Tail Probability of Desired BS
Distance from MS}

From the proof of Lemma \ref{lem:StochasticEqLemma}, the arrangement
of the BSs of the $k^{\mathrm{th}}$ tier in the equivalent single-tier
network is according to a non-homogeneous Poisson point process with
BS density function $\lambda^{\left(k\right)}\left(r\right)=\lambda_{k}\frac{2\pi}{\varepsilon_{k}}\left(P_{k}B_{k}\right)^{\frac{2}{\varepsilon_{k}}}\mathbb{E}\Psi_{k}^{\frac{2}{\varepsilon_{k}}}r^{\frac{2}{\varepsilon_{k}}-1},\ r\ge0,$
and $\tilde{R}_{1}^{\left(k\right)}$ represents the distance of the
nearest BS of this random process. Using the properties of the Poisson
point process, $\mathbb{P}\left(\left\{ \tilde{R}_{1}^{\left(k\right)}>r\right\} \right)=\mathbb{P}\left(\left\{ N^{\left(k\right)}\left(\left[0,r\right]\right)=0\right\} \right),$
where $N^{\left(k\right)}\left(\left[0,r\right]\right)$ represents
the average number of BSs in the interval $\left[0,r\right]$ placed
according to the Poisson point process with density $\lambda^{\left(k\right)}\left(r\right),$
and this is equal to $\left(\ref{eq:nearestKthTierBSDistance}\right).$

\subsection{\label{sub:proofDesiredBSTierProbLemma}Desired BS Tier Probability
Lemma}

The following sequence of equations provides the proof. \begin{eqnarray*}
\mathbb{P}\left(\left\{ I=k\right\} \right) & \overset{\left(a\right)}{=} & \mathbb{P}\left(\bigcap_{l=1,\ l\ne k}^{K}\left\{ \tilde{R}_{1}^{\left(k\right)}<\tilde{R}_{1}^{\left(l\right)}\right\} \right)\\
 & \overset{\left(b\right)}{=} & \mathbb{E}_{\tilde{R}_{1}^{\left(k\right)}}\left[\prod_{l=1,\ l\ne k}^{K}\mathbb{P}\left(\left\{ \tilde{R}_{1}^{\left(l\right)}>\tilde{R}_{1}^{\left(k\right)}\right\} \right)\right]\\
 & \overset{\left(c\right)}{=} & \int_{r=0}^{\infty}\lambda_{k}\frac{2\pi}{\varepsilon_{k}}\left(P_{k}B_{k}\right)^{\frac{2}{\varepsilon_{k}}}\mathbb{E}\Psi_{k}^{\frac{2}{\varepsilon_{k}}}r^{\frac{2}{\varepsilon_{k}}-1}\times\\
 &  & \exp\left(-\sum_{l=1}^{K}\lambda_{l}\pi\left(P_{l}B_{l}\right)^{\frac{2}{\varepsilon_{l}}}\mathbb{E}\Psi_{l}^{\frac{2}{\varepsilon_{l}}}r^{\frac{2}{\varepsilon_{l}}}\right)dr,\end{eqnarray*}
where $\left(a\right)$ is obtained by noting that the desired BS
belonging to the $k^{\mathrm{th}}$ tier is closer to the origin than
the nearest BS (to origin) of the rest of the tiers, $\left(b\right)$
is due to the fact that $\left\{ R_{1}^{\left(l\right)}\right\} _{l=1}^{K}$
are independent random variables, and $\mathbb{E}_{R_{1}^{\left(k\right)}}\left[\cdot\right]$
represents the expectation with respect to the random variable $R_{1}^{\left(k\right)},$
$\left(c\right)$ is obtained by using Theorem \ref{thm:nearestKthTierBSDistance}-Equation
$\left(\ref{eq:nearestKthTierBSDistance}\right)$ to compute the probability
of the event in $\left(b\right)$ and to obtain the probability density
function (p.d.f.) of $\tilde{R}_{1}^{\left(k\right)}$ to evaluate
the expectation, and finally, $\left(\ref{eq:TierProbabilityGeneral}\right)$
is obtained by simplifying $\left(c\right).$ 

When $\left\{ \varepsilon_{k}\right\} _{k=1}^{K}=\varepsilon,$ the
integral in $\left(\ref{eq:TierProbabilityGeneral}\right)$ simplifies
to the form $\int_{t=0}^{\infty}\mathrm{e}^{-\alpha t}dt=\frac{1}{\alpha},$
which is rewritten in $\left(\ref{eq:TierProbabilitySpecialCase}\right).$%
\begin{figure*}
\vspace{-0.2in}
\lyxline{\normalsize}\begin{eqnarray}
\mathbb{P}\left(\left\{ \tilde{R}_{1}>r,\ I=k\right\} \right) & \overset{\left(a\right)}{=} & \mathbb{P}\left(\left\{ \tilde{R}_{1}^{\left(k\right)}>r\right\} \bigcap\bigcap_{l=1,\ l\ne k}^{K}\left\{ \tilde{R}_{1}^{\left(l\right)}>\tilde{R}_{1}^{\left(k\right)}\right\} \right)\nonumber \\
 & \overset{\left(b\right)}{=} & \mathbb{E}_{\tilde{R}_{1}^{\left(k\right)}}\left[\mathcal{I}\left(\left\{ \tilde{R}_{1}^{\left(k\right)}>r\right\} \right)\cdot\mathbb{P}\left(\left.\bigcap_{l=1,\ l\ne k}^{K}\left\{ \tilde{R}_{1}^{\left(l\right)}>\tilde{R}_{1}^{\left(k\right)}\right\} \right|\tilde{R}_{1}^{\left(k\right)}\right)\right]\nonumber \\
 & \overset{\left(c\right)}{=} & \mathbb{E}_{\tilde{R}_{1}^{\left(k\right)}}\left[\mathcal{I}\left(\left\{ \tilde{R}_{1}^{\left(k\right)}>r\right\} \right)\cdot\prod_{l=1,\ l\ne k}^{K}\mathbb{P}\left(\left.\left\{ \tilde{R}_{1}^{\left(l\right)}>\tilde{R}_{1}^{\left(k\right)}\right\} \right|\tilde{R}_{1}^{\left(k\right)}\right)\right]\nonumber \\
 & \overset{\left(d\right)}{=} & \mathbb{E}_{\tilde{R}_{1}^{\left(k\right)}}\left[\mathcal{I}\left(\left\{ \tilde{R}_{1}^{\left(k\right)}>r\right\} \right)\cdot\exp\left(-\sum_{l=1,\ l\ne k}^{K}\lambda_{l}\pi\left(P_{l}B_{l}\tilde{R}_{1}^{\left(k\right)}\right)^{\frac{2}{\varepsilon_{l}}}\mathbb{E}\Psi_{l}^{\frac{2}{\varepsilon_{l}}}\right)\right]\nonumber \\
 & \overset{\left(e\right)}{=} & \int_{t=r}^{\infty}\lambda_{k}\frac{2\pi}{\varepsilon_{k}}\left(P_{k}B_{k}\right)^{\frac{2}{\varepsilon_{k}}}\mathbb{E}\Psi_{k}^{\frac{2}{\varepsilon_{k}}}t^{\frac{2}{\varepsilon_{k}}-1}\cdot\exp\left(-\sum_{l=1}^{K}\lambda_{l}\pi\left(P_{l}B_{l}t\right)^{\frac{2}{\varepsilon_{l}}}\mathbb{E}\Psi_{l}^{\frac{2}{\varepsilon_{l}}}\right)dt,\label{eq:proofRkgivenTierTailProb}\\
\mathbb{P}_{coverage}^{open-access} & = & \mathbb{P}\left(\left\{ SINR>\beta_{I}\right\} \right)\overset{\left(a\right)}{=}\mathbb{E}_{I,\tilde{R}_{1}}\left[\mathbb{P}\left(\left.\left\{ \frac{\sum_{m=2}^{\infty}\tilde{R}_{m}^{-1}+\eta}{\tilde{R}_{1}^{-1}}<\beta_{I}^{-1}\right\} \right|I,\ \tilde{R}_{1}\right)\right]\nonumber \\
 & \overset{\left(b\right)}{=} & \mathbb{E}_{I,\tilde{R}_{1}}\left[\int_{x=0}^{\beta_{I}^{-1}}\int_{\omega=-\infty}^{\infty}\Phi_{\left.\sum_{m=2}^{\infty}\tilde{R}_{m}^{-1}+\eta\right|\tilde{R}_{1}}\left(\omega\tilde{R}_{1}\right)\frac{\mathrm{e}^{-i\omega}}{2\pi}d\omega dx\right]\nonumber \\
 & \overset{\left(c\right)}{=} & \mathbb{E}_{I}\left[\int_{\omega=-\infty}^{\infty}\mathbb{E}_{\left.\tilde{R}_{1}^{-1}\right|I}\left[\Phi_{\left.\sum_{m=2}^{\infty}\tilde{R}_{m}^{-1}\right|\tilde{R}_{1}}\left(\omega\tilde{R}_{1}\right)\right]\cdot\frac{\mathrm{e}^{i\omega\eta\tilde{R}_{1}}\left(1-\mathrm{e}^{-\frac{i\omega}{\beta_{I}}}\right)}{2\pi i\omega}d\omega dx\right],\label{eq:proofCovProbTheorem}\end{eqnarray}
\lyxline{\normalsize}\vspace{-0.1in}

\end{figure*}

\subsection{\label{sub:proofPDFRkGivenTier}Proof for the p.d.f. of the Desired
BS Distance given Tier}

We first evaluate the probability of the event $\left\{ \tilde{R}_{1}>r,\ I=k\right\} .$
The steps are given in $\left(\ref{eq:proofRkgivenTierTailProb}\right),$
where $\left(a\right)$ is obtained by noting that the desired BS
belonging to the $k^{\mathrm{th}}$ tier is closer to the origin than
the nearest BS (to origin) of the rest of the tiers, $\left(b\right)$
rewrites joint probability in $\left(a\right)$ in terms of the product
of the marginal and the conditional probabilities, $\left(c\right)$
is obtained by noting that the $\tilde{R}_{1}^{\left(l\right)},\ \forall\ l=1,\ 2,\cdots,\ K$
are independent random variables, $\left(d\right)$ is obtained by
substituting for the tail probability events in $\left(c\right)$
using Theorem \ref{thm:nearestKthTierBSDistance}- Equation $\left(\ref{eq:nearestKthTierBSDistance}\right),$
and finally, $\left(e\right)$ is obtained by first evaluating the
p.d.f. of $\tilde{R}_{1}^{\left(k\right)}$ using $\left(\ref{eq:nearestKthTierBSDistance}\right)$
and expressing the expectation as an integration.

As a result, $\mathbb{P}\left(\left\{ \left.\tilde{R}_{1}>r\right|I=k\right\} \right)=\frac{\mathbb{P}\left(\left\{ \tilde{R}_{1}>r,\ I=k\right\} \right)}{\mathbb{P}\left(\left\{ I=k\right\} \right)},$
and finally, $\left(\ref{eq:pdfServingBSgivenTier}\right)$ is obtained
by noting that $f_{\left.\tilde{R}_{1}\right|I}\left(\left.r\right|k\right)=-\frac{d}{dr}\mathbb{P}\left(\left\{ \left.\tilde{R}_{1}>r\right|I=k\right\} \right).$\vspace{-0.1in}

\subsection{\label{sub:ProofCovProbTheorem}Proof for the Coverage Probability
Theorem}

The sequence of equations in $\left(\ref{eq:proofCovProbTheorem}\right),$
where $\left(a\right)$ is obtained by using Lemma \ref{lem:StochasticEqLemma},
and basic conditional probability properties, $\left(b\right)$ is
obtained by using \cite[Theorem 1]{Madhusudhanan2010a}, using the
BS density function specified in Lemma \ref{lem:StochasticEqLemma}
for $\lambda\left(r\right)$ and 1 for the path-loss exponent, and
$\left(c\right)$ is obtained by exchanging the order of integrations
in $\left(b\right),$ which is valid since the integrals are convergent.
Upon simplifying, we get \begin{eqnarray*}
 &  & \Phi_{\left.\sum_{m=2}^{\infty}\tilde{R}_{m}^{-1}\right|\tilde{R}_{1}}\left(\omega\tilde{R}_{1}\right)=\exp\left(\sum_{l=1}^{K}\lambda_{l}\mathbb{E}\Psi_{l}^{\frac{2}{\varepsilon_{l}}}\pi\left(P_{l}B_{l}\tilde{R}_{1}\right)^{\frac{2}{\varepsilon_{l}}}\right.\\
 &  & \left.\times\left(1-_{1}F_{1}\left(-\frac{2}{\varepsilon_{l}};1-\frac{2}{\varepsilon_{l}};i\omega\right)\right)\right).\end{eqnarray*}
 Further, by evaluating the expectation in $\left(c\right)$ by using
Lemma \ref{lem:pdfServingBSGivenTier} - Equation $\left(\ref{eq:pdfServingBSgivenTier}\right),$
and simplifying, we get $\left(\ref{eq:coverageProbabilityGeneral}\right).$

\subsection{\label{sub:Simulation-Method}Simulation Method}

The $k^{\mathrm{th}}$ tier of the heterogeneous network with $K$
tiers is identified by the following set of system parameters: $\left(\lambda_{k},\ P_{k},\ B_{k},\ \Psi_{k},\ \varepsilon_{k},\ \beta_{k}\right),$
where the symbols have all been defined in Section \ref{sec:modelreview},
and $k=1,\ 2,\cdots,\ K,$ where $K$ is the total number of tiers.
Now we illustrate the steps for simulating the heterogeneous network
in order to obtain the SINR distribution and the coverage probability
in the open-access cell association scheme. Assuming the MS to be
at the origin, a single trial of heterogeneous cellular network arrangement
in a cellular area with $R_{B}$ as the boundary radius involves:

1) Generating the random numbers $N_{k}\sim\mathrm{Poisson}\left(\lambda_{k}\pi R_{B}^{2}\right),$
which is the number of BSs of the $k^{\mathrm{th}}$ that will be
deployed in the trial.

2) Generating $N_{k}$ random variables according to a uniform distribution
in the circular region of area $\pi R_{B}^{2},$ which represents
the location of the $k^{\mathrm{th}}$ tier BSs corresponding to the
trial.

3) Computing the received power at the MS at the origin and computing
the SINR as the ratio of the maximum of the received powers to the
difference of the sum of all the received powers and the maximum received
power.

4) Record the index $I$ which corresponds to the tier to which the
desired BS belongs, for the tier. 

Repeat the same procedure T times. Typically, T is at least 50000.
After this, we have an array containing the instantaneous SINRs and
the tiers to which the desired BSs belonged, corresponding to the
T trials. The tail probability of SINR at a certain point, say $\eta,$
is given by $\frac{\left\{ \mbox{\# of trials where SINR >}\ \eta\right\} }{T},$
and the coverage probability of the MS in the heterogeneous network
is given by $\sum_{k=1}^{K}\frac{\left\{ \mbox{\# of trials where}\ I=k\ \mathrm{and}\ SINR>\beta_{k}\right\} }{T}$. 
\end{document}